\documentclass[%
reprint, 
superscriptaddress, 
amsmath, amssymb, 
aps, 
pra, 
]{revtex4-2}

\usepackage{graphicx}
\usepackage{dcolumn}
\usepackage{bm}
\usepackage{braket}
\usepackage{amsthm}
\usepackage{dsfont}
\usepackage{hyperref}

\newtheorem{Thm}{Theorem}
\newtheorem{Lem}[Thm]{Lemma}
\newtheorem{Prop}[Thm]{Proposition}

\theoremstyle{definition}

\begin{document}

\title{Unconditional quantum advantage from a 
two-round CHSH problem in one dimension}

\author{Yonghae Lee}
 \email{yonghaelee@kangwon.ac.kr}
 \affiliation{Department of Liberal Studies, Kangwon National University, Samcheok 25913, Korea}
\author{Jeonghyeon Shin}
 \email{jeonghyeon.shin@kist.re.kr}
 \affiliation{Center for Quantum Technology, Korea Institute of Science and Technology (KIST), Seoul 02792, Korea}
 \affiliation{Department of Mathematics and Research Institute for Basic Sciences, Kyung Hee University, Seoul 02447, Korea}
\author{Soojoon Lee}%
 \email{level@khu.ac.kr}
\affiliation{Department of Mathematics and Research Institute for Basic Sciences, Kyung Hee University, Seoul 02447, Korea}%
\affiliation{School of Computational Sciences, Korea Institute for Advanced Study, Seoul 02455, Korea}%

\date{\today}

\begin{abstract}
We introduce a relation problem constructed from the Clauser--Horne--Shimony--Holt (CHSH) game, which we call the two-round one-dimensional CHSH problem.
Its two-round structure ensures that the CHSH questions are supplied only after the relevant Pauli-frame data have been fixed, thereby ruling out a simple classical strategy that solves the corresponding problem perfectly when all inputs are supplied simultaneously.
We construct a quantum circuit on $2N$ qubits that uses only adjacent two-qubit gates, has operational depth at most eight, and achieves the optimal quantum success probability of CHSH, which is strictly smaller than one.
We prove that, for every fixed $0\leq\delta<(\sqrt{2}-1)/4$, any randomized classical circuit with fixed wiring and bounded gate fan-in that achieves an average success probability of at least $(2+\sqrt{2})/4-\delta$ requires depth $\Omega(\log N)$ after the questions of the second round are supplied.
This yields an unconditional separation even though the quantum circuit is restricted to a one-dimensional geometry, whereas the classical circuit has no geometric locality restriction.
The result shows that perfect quantum success is not necessary for unconditional quantum advantage with shallow circuits.
\end{abstract}

\maketitle


\section{Introduction} \label{sec:Introduction}

The demonstration of quantum advantage is a central objective of quantum information science.
A rigorous separation identifies computational tasks for which quantum resources provide performance that cannot be reproduced within a specified classical model.
Quantum algorithms and sampling experiments have provided strong evidence for such an advantage~\cite{Shor1994, Bremner2011, Arute2019}.
However, the corresponding statements about classical hardness generally depend on unresolved assumptions in computational complexity.

Unconditional separations play a complementary role by establishing a provable advantage under explicit restrictions on the computational models.
The term \emph{unconditional} means that the classical lower bound does not rely on any unproved assumption about computational hardness.
It does not mean that the comparison places no restrictions on classical computation.
In this work, the relevant restrictions concern gate fan-in, circuit depth, fixed circuit wiring, and the temporal order in which inputs become available.
Shallow circuits provide a natural setting for such a comparison because their depth and gate fan-in limit the propagation of information, whereas quantum entanglement can produce correlations inaccessible to classical circuits under the same depth restriction~\cite{Bravyi2018, LeGall2019, Bravyi2020}.

The CHSH game is a basic example of quantum nonlocality in which the optimal classical and quantum winning probabilities are $3/4$ and $(2+\sqrt{2})/4$, respectively~\cite{CHSH1969, Tsirelson1980, Cleve2004, Brunner2014}.
In this work, we introduce the two-round one-dimensional CHSH (2R1D-CHSH) problem, a relation problem constructed from the CHSH game.

The division into two rounds is an essential part of the problem.
If the CHSH questions were available when the transcript of the first round was produced, a classical circuit of depth one could encode those questions in the Pauli-frame data and satisfy the resulting relation with probability one.
The temporal ordering rules out this strategy by fixing the Pauli frame before the questions are supplied.
For any selected pair whose answers cannot depend on the opposite questions, the resulting strategy is then subject to the classical CHSH bound.
A separate lightcone argument shows that only a small fraction of selected pairs can violate this independence condition when the classical circuit has small depth and bounded fan-in.

Previous two-round tasks for shallow quantum advantage use interaction to divide the simulation of a Clifford computation and thereby obtain hardness against stronger classical circuit classes~\cite{Grier2020, Grier2021}.
More recently, two-round interactive problems with perfect completeness have been used to establish a depth hierarchy within shallow quantum circuits~\cite{Hsieh2026}.
In the 2R1D-CHSH problem, the two rounds instead enforce the causal independence needed to retain the classical and quantum values of CHSH.
Accordingly, the resulting separation is expressed through their average success probabilities rather than perfect success.

Our main result establishes an unconditional separation in circuit depth for the 2R1D-CHSH problem.
We construct a quantum circuit on $2N$ qubits that uses gates acting only on adjacent qubits, has operational depth at most eight, and achieves the optimal quantum success probability of the CHSH game.
We also derive an upper bound on the average success probability of randomized classical circuits with bounded fan-in and show that every classical circuit family of fixed depth is asymptotically limited by the classical value of the CHSH game.
Consequently, any such classical circuit that approximates the quantum performance within a fixed error smaller than the gap between the classical and quantum values must have depth $\Omega(\log N)$.
This separation holds even though the quantum circuit is restricted to a one-dimensional geometry, whereas the classical circuit may have arbitrary size and fan-out and is not subject to any geometric locality restriction.

The present result differs from previous shallow circuit separations in two main respects.
First, the closely related one-dimensional constructions based on the Magic Square and Magic Pentagram games use perfect quantum strategies and require $4N$ and $6N$ qubits, respectively~\cite{Bravyi2020, Han2022}.
The 2R1D-CHSH construction instead begins with a game whose quantum value is strictly smaller than one and uses only $2N$ qubits.
Second, CHSH has previously appeared as an illustrative example of an advantage enabled by two stages~\cite{Caha2024}, and nonlocal games including CHSH can be converted into interactive tests of quantum advantage using cryptographic assumptions~\cite{Kalai2023}.
These results do not give a direct and unconditional circuit depth separation for a relation problem constructed from CHSH.

The present construction gives an explicit relation problem derived directly from CHSH, together with a quantitative upper bound on the average success probability of classical circuits.
Unlike constructions that rely on perfect quantum strategies, it converts the strict but probabilistic advantage of CHSH into an asymptotic separation in circuit depth.
Thus, quantum pseudo-telepathy~\cite{Brassard2005} is not required for this construction.
The essential mechanism is the temporal ordering that prevents the questions from influencing the previously reported Pauli-frame data, combined with the limited propagation of those questions through a shallow classical circuit.

The remainder of this paper is organized as follows.
In Sec.~\ref{sec:CHSHg}, we state the CHSH conventions and the properties of its Pauli-frame generalization needed below.
In Sec.~\ref{sec:2R1D}, we formulate the 2R1D-CHSH problem and specify its input distribution, temporal structure, and the condition defining a valid output.
In Sec.~\ref{sec:Quantum}, we construct the constant-depth quantum circuit and determine its average success probability.
In Sec.~\ref{sec:Classical}, we establish an upper bound on the average success probability of randomized classical circuits with bounded fan-in.
In Sec.~\ref{sec:DepthSeparation}, we combine the quantum and classical results to obtain the circuit depth separation.
In Sec.~\ref{sec:Remarks}, we explain why the two-round structure is required and compare the qubit requirements with those of previous constructions.
Finally, Sec.~\ref{sec:Conclusion} concludes the paper and discusses directions for further research.

\section{The CHSH Game and Its Pauli-Frame Generalization} \label{sec:CHSHg}

In this section, we specify the CHSH conventions and the Pauli-frame properties used in the construction and analysis below.
For a nonlocal game $G$, we write $\omega_{\mathrm{c}}(G)$ and $\omega_{\mathrm{q}}(G)$ for its optimal classical and quantum winning probabilities, respectively.

In the CHSH game, a verifier sends independent and uniformly distributed bits $x,y\in\{0,1\}$ to Alice and Bob, respectively.
They return bits $a, b\in\{0,1\}$ without communication after receiving their questions and win when
\begin{equation} \label{eq:CHSH-winning-condition}
a \oplus b = xy, 
\end{equation}
where $ \oplus $ denotes addition modulo two.
Classical strategies may use shared randomness, whereas quantum strategies may also use a shared entangled state prepared before the questions are chosen.
The optimal winning probabilities are~\cite{CHSH1969, Tsirelson1980, Cleve2004}
\begin{equation} \label{eq:CHSH-values}
\omega_{\mathrm{c}}(\mathrm{CHSH}) = \frac{3}{4}, \quad
\omega_{\mathrm{q}}(\mathrm{CHSH}) = \frac{2+\sqrt{2}}{4}.
\end{equation}

An optimal quantum strategy uses the Bell state
\begin{equation} \label{eq:Phi-plus}
\ket{\Phi^+} = \frac{\ket{00}+\ket{11}}{\sqrt{2}}.
\end{equation}
Alice and Bob measure $\mathcal{A}_x$ and $\mathcal{B}_y$, respectively.
Alice's observables are
\begin{equation} \label{eq:CHSH-Alice-observables}
\mathcal{A}_0 = Z, \quad
\mathcal{A}_1 = X, 
\end{equation}
and Bob's observables are
\begin{equation} \label{eq:CHSH-Bob-observables}
\mathcal{B}_0 = \frac{Z+X}{\sqrt{2}}, \quad
\mathcal{B}_1 = \frac{Z-X}{\sqrt{2}}.
\end{equation}
The symbols $X$ and $Z$ denote the Pauli matrices.
The measurement eigenvalues $+1$ and $-1$ are reported as the bits $0$ and $1$, respectively.

A Pauli frame records the Pauli operators relating a physical state to a reference state, so that their effects can be accounted for without applying a physical correction~\cite{Knill2005, Chamberland2018}.
For $\sigma,\tau\in\{0,1\}$, we use the convention~\cite{Bennett1993}
\begin{equation} \label{eq:framed-Bell-state}
\ket{\Phi_{\sigma,\tau}} = \left( Z^\tau X^\sigma \otimes \mathbf{1} \right)\ket{\Phi^+}, 
\end{equation}
where $\mathbf{1}$ is the identity on Bob's qubit.
Thus, $\sigma$ is the $X$-frame bit and $\tau$ is the $Z$-frame bit.

For a fixed frame $(\sigma,\tau)$, the Pauli-frame CHSH (PF-CHSH) game has the same questions, distribution, answer alphabets, and communication restriction as CHSH, but its winning condition is
\begin{equation} \label{eq:PF-CHSH-bit-condition}
a \oplus b = xy \oplus (1-x)\sigma \oplus x\tau.
\end{equation}
The frame is fixed and known to the players before they choose their strategy.
We suppress its values in the notation $\mathrm{PF\text{-}CHSH}$.
The local relabeling $a' = a \oplus (1-x)\sigma \oplus x\tau$ is invertible and converts this condition into $a' \oplus b = xy$.
Hence, for every fixed frame,
\begin{equation} \label{eq:PF-CHSH-values}
\begin{aligned}
\omega_{\mathrm{c}}(\mathrm{PF\text{-}CHSH}) & = \omega_{\mathrm{c}}(\mathrm{CHSH}) = \frac{3}{4}, \\
\omega_{\mathrm{q}}(\mathrm{PF\text{-}CHSH}) & = \omega_{\mathrm{q}}(\mathrm{CHSH}) = \frac{2+\sqrt{2}}{4}.
\end{aligned}
\end{equation}

The quantum value is attained by measurements of the same observables $\mathcal{A}_x$ and $\mathcal{B}_y$ on
$\ket{\Phi_{\sigma,\tau}}$.
For this strategy, every question pair satisfies
\begin{equation} \label{eq:PF-CHSH-pointwise-success}
\Pr \left[ a \oplus b = xy \oplus (1-x)\sigma \oplus x\tau | x,y \right]
 = \frac{2+\sqrt{2}}{4}.
\end{equation}
The measurements do not depend on the frame bits, and no Pauli correction is applied to the selected endpoint qubits.
In the circuit construction below, the frame enters only the verification condition.
Appendix~\ref{app:CHSH-details} gives the classical bound and the correlation calculation underlying these statements.

\section{The Two-Round One-Dimensional CHSH Problem} \label{sec:2R1D}

In this section, we formally define the 2R1D-CHSH problem, which is based on the PF-CHSH game introduced in Sec.~\ref{sec:CHSHg}.
We specify its input-output structure over two rounds and the corresponding relation, prescribe the uniform and independent input distribution used throughout this work, and define the average success probability with respect to this distribution.
Unlike a conventional relation problem completed in a single round, the input and output of the 2R1D-CHSH problem are supplied and produced, respectively, over two rounds.

The 2R1D-CHSH problem has separate input and output sets for the two rounds.
The output of the first round must be produced before the input to the second round is supplied.
Whether the complete input-output tuple satisfies the relation is determined only after the output of the second round has been produced.

For each integer $N\geq 2$, let $\mathcal{I}_N^{\mathrm{1st}}$ and $\mathcal{O}_N^{\mathrm{1st}}$ denote the input and output sets of the first round, respectively.
Each input to the first round consists of two $N$-bit strings.
The set of all allowed inputs to the first round is
\begin{equation} \label{eq:1Iset}
\mathcal{I}_N^{\mathrm{1st}} = \left\{ (\mathbf{p}[j], \mathbf{q}[k]) : 1 \leq j<k \leq N \right\}, 
\end{equation}
where $\mathbf{p}[j]$ and $\mathbf{q}[k]$ are the $N$-bit strings with a single nonzero component, given by
\begin{equation}
p_i = 
\begin{cases}
1, & i = j, \\
0, & i\neq j, 
\end{cases}
\quad
q_i = 
\begin{cases}
1, & i = k, \\
0, & i\neq k.
\end{cases}
\end{equation}
Thus, $\mathbf{p}[j]$ and $\mathbf{q}[k]$ specify the two selected positions $j$ and $k$, respectively.
The set of all possible outputs of the first round is
\begin{equation}
\mathcal{O}_N^{\mathrm{1st}} = \{0,1\}^N \times \{0,1\}^N.
\end{equation}
We denote a generic output of the first round by
\begin{equation}
(\mathbf{s}, \mathbf{t}) \in \mathcal{O}_N^{\mathrm{1st}}, 
\end{equation}
where the two component strings are written as
\begin{equation}
\mathbf{s} = (s_1, \ldots, s_N), \quad \mathbf{t} = (t_1, \ldots, t_N), 
\end{equation}
with $s_i, t_i\in\{0,1\}$ for every $i\in\{1, \ldots, N\}$.

Once the output of the first round has been produced, the second round begins.
Let $\mathcal{I}_N^{\mathrm{2nd}}$ and $\mathcal{O}_N^{\mathrm{2nd}}$ denote the input and output sets of the second round, respectively.
Both sets have the same form:
\begin{equation}
\mathcal{I}_N^{\mathrm{2nd}} = \mathcal{O}_N^{\mathrm{2nd}} = \{0,1\}^N\times\{0,1\}^N.
\end{equation}
We denote a generic input and output of the second round by
\begin{equation}
(\mathbf{x}, \mathbf{y}) \in \mathcal{I}_N^{\mathrm{2nd}}, \quad
(\mathbf{a}, \mathbf{b}) \in \mathcal{O}_N^{\mathrm{2nd}}, 
\end{equation}
where the four component strings are written as
\begin{align}
\mathbf{x}& = (x_1, \ldots, x_N), &\mathbf{y}& = (y_1, \ldots, y_N), \\
\mathbf{a}& = (a_1, \ldots, a_N), &\mathbf{b}& = (b_1, \ldots, b_N).
\end{align}
Each component satisfies $x_i, y_i, a_i, b_i\in\{0,1\}$ for every $i\in\{1, \ldots, N\}$.

We now define the relation associated with the 2R1D-CHSH problem. It is a subset
\begin{equation}
\mathcal{R}_N
\subseteq
\mathcal{I}_N^{\mathrm{1st}} \times
\mathcal{O}_N^{\mathrm{1st}} \times
\mathcal{I}_N^{\mathrm{2nd}} \times
\mathcal{O}_N^{\mathrm{2nd}}.
\end{equation}
For each selected pair $1\leq j<k\leq N$, we define the two associated Pauli-frame bits by
\begin{equation} \label{eq:2R1D-frame-bits}
\sigma_{j,k}(\mathbf{s}) = \bigoplus_{i = j+1}^{k} s_i, \quad
 \tau_{j,k}(\mathbf{t}) = \bigoplus_{i = j}^{k-1} t_i.
\end{equation}
In the quantum construction of Sec.~\ref{sec:Quantum}, these parities are the $X$- and $Z$-frame bits, respectively, of the Bell state shared by the selected registers.
Then, a complete input-output tuple belongs to $\mathcal{R}_N$ if and only if the input and output bits at the selected positions $j$ and $k$ satisfy the following PF-CHSH parity condition:
\begin{widetext}
\begin{equation} \label{eq:Relation}
\left( (\mathbf{p}[j], \mathbf{q}[k]), (\mathbf{s}, \mathbf{t}), (\mathbf{x}, \mathbf{y}), (\mathbf{a}, \mathbf{b}) \right)
\in \mathcal{R}_N
\quad \Longleftrightarrow \quad
a_j \oplus b_k = x_jy_k \oplus (1 - x_j)\sigma_{j,k}(\mathbf{s}) \oplus x_j\tau_{j,k}(\mathbf{t}).
\end{equation}
\end{widetext}

Only the four bits $x_j, y_k, a_j, b_k$ from the second round appear in the verification predicate.
The bits $x_j, y_k$ serve as the two PF-CHSH questions, whereas $a_j, b_k$ serve as the corresponding answers.
Together with the Pauli-frame bits $\sigma_{j,k}(\mathbf{s})$ and $\tau_{j,k}(\mathbf{t})$, these four bits satisfy precisely the PF-CHSH condition in Eq.~(\ref{eq:PF-CHSH-bit-condition}).

Equation~(\ref{eq:Relation}) shows that only two of the $2N$ input bits and two of the $2N$ output bits of the second round enter the verification predicate.
The remaining $2N-2$ input bits and $2N-2$ output bits may therefore appear redundant from the viewpoint of the static relation alone.
Nevertheless, we retain the full $2N$-bit representation because it provides a local encoding compatible with the circuit model considered below.

In particular, the quantum strategy will be implemented by a single circuit executed in two stages, corresponding to the two rounds of the 2R1D-CHSH problem.
The portion of the circuit used in the first round must be connected, at the selected positions $j$ and $k$, to the portion used in the second round.
Since these positions are chosen at random rather than fixed in the circuit architecture, the circuit must provide input and output registers at every possible position.
The $2N$-bit representation accomplishes this by assigning an input bit and an output bit to every site in each of the two $N$-site registers.
The relation then selects the components $x_j, y_k, a_j, b_k$ associated with the realized pair $(j,k)$.
This representation therefore allows a single circuit family to handle every possible choice of the selected positions without changing its architecture.

The locations of these input and output registers are part of the computational model.
In particular, the input bits at the selected positions are not supplied through additional terminals chosen after the first round.
Any selection or routing of these bits must be performed by the circuit itself.

The order of the two rounds is an essential part of the 2R1D-CHSH problem.
In particular, the input $(\mathbf{x}, \mathbf{y})$ to the second round is not available when the output $(\mathbf{s}, \mathbf{t})$ of the first round is produced.
Once produced, this output is fixed and cannot be modified after the input to the second round becomes available.
The necessity of this temporal condition, and what can go wrong if it is removed, will be discussed in Sec.~\ref{sec:Why}.

We complete the specification of the 2R1D-CHSH problem by defining its input distribution.
Throughout this work, the input to each round is sampled uniformly, and the input to the second round is independent of both the input and the output of the first round. 

In the probabilistic formulation below, uppercase symbols denote random variables and the corresponding lowercase symbols denote their realizations, while boldface is used for bit strings.
Thus, for example, $J = j$, $\mathbf{X} = \mathbf{x}$, and $A_j = a_j$ denote realization events.

The pair $(J,K)$ is chosen uniformly at random among all pairs of integers satisfying $1\leq J<K\leq N$.
Its probability mass function is therefore
\begin{equation}
\Pr[J=j,K=k] = \frac{1}{\binom{N}{2}}
\end{equation}
for every $1\leq j<k\leq N$.

After the output of the first round has been produced, the input strings $\mathbf{X}$ and $\mathbf{Y}$ are sampled independently and uniformly from $\{0,1\}^N$.
Their joint distribution is also independent of $J,K, \mathbf{S}, \mathbf{T}$.
More explicitly, this joint independence and uniformity are expressed by
\begin{equation} \label{eq:2R1D-input-independence}
\Pr \left[ \mathbf{X} = \mathbf{x}, \mathbf{Y} = \mathbf{y} | J = j, K = k, \mathbf{S} = \mathbf{s}, \mathbf{T} = \mathbf{t} \right]
 = \frac{1}{2^{2N}}, 
\end{equation}
whenever the conditioning event has nonzero probability. In particular, the two input bits that enter the relation satisfy
\begin{equation}
\Pr \left[ X_j = x, Y_k = y | J = j, K = k, \mathbf{S} = \mathbf{s}, \mathbf{T} = \mathbf{t} \right]
 = \frac{1}{4}
\end{equation}
for all $x,y\in\{0,1\}$.
Thus, the Pauli-frame bits fixed in the first round cannot influence the distribution of the PF-CHSH questions supplied in the second round.

The verifier generates these inputs using randomness independent of the strategy.
In particular, the input strings of the second round are independent of the complete internal state retained after the first round, including any classical or quantum registers, and of any private or shared randomness used by the strategy.

Under this input distribution, the average success probability of a two-round strategy $\mathsf{S}$ is defined as
\begin{equation} \label{eq:2R1D-average-success}
\overline{p}_{\mathrm{succ}}(\mathsf{S})
 = \frac{1}{4\binom{N}{2}} \sum_{1\leq j<k\leq N} \sum_{x,y\in\{0,1\}} p_{\mathrm{succ}}^{\mathsf{S}}(j,k;x,y).
\end{equation}
The quantity $p_{\mathrm{succ}}^{\mathsf S}(j,k;x,y)$ denotes the conditional probability that the output produced by the strategy $\mathsf S$ satisfies the relation when the selected positions are $j,k$ and the relevant input bits are $x,y$.
More precisely, it is defined as follows:
\begin{widetext}
\begin{equation} \label{eq:2R1D-conditional-success}
p_{\mathrm{succ}}^{\mathsf{S}}(j,k;x,y)
: = \Pr\nolimits_{\mathsf{S}}
\left[
A_j \oplus B_k = xy \oplus (1-x)\sigma_{j,k}(\mathbf{S}) \oplus x\tau_{j,k}(\mathbf{T})
| J = j, K = k, X_j = x, Y_k = y
\right].
\end{equation}
\end{widetext}
For any event $E$, the notation $\Pr\nolimits_{\mathsf{S}}[E]$ denotes the probability of $E$ in the experiment in which the prescribed inputs are sampled and the strategy $\mathsf{S}$ is executed.
Accordingly, $\Pr\nolimits_{\mathsf{S}}[E| C]$ denotes the corresponding conditional probability given an event $C$.
This probability includes any internal randomness of a classical strategy and any measurement randomness of a quantum strategy.

The variables $A_j$ and $B_k$ appear in Eq.~(\ref{eq:2R1D-conditional-success}) because the answers produced by the strategy $\mathsf{S}$ need not be deterministic, even after the relevant input bits have been fixed.
Similarly, $\sigma_{j,k}(\mathbf{S})$ and $\tau_{j,k}(\mathbf{T})$ are random variables determined by the random output of the first round.
Since the conditioning fixes only $X_j = x$ and $Y_k = y$, the conditional probability also averages over the remaining $2N-2$ input bits of the second round, the output of the first round, and all randomness inherent in the strategy.

The factor $1/\binom{N}{2}$ in Eq.~(\ref{eq:2R1D-average-success}) averages over the possible selected pairs $j<k$, whereas the factor $1/4$ averages over the four possible values of the relevant input pair $(x,y)$.
There is no separate success condition for the first round.
Its output instead fixes the Pauli-frame bits that enter the final condition evaluated after the second round.
Thus, $\overline{p}_{\mathrm{succ}}(\mathsf{S})$ evaluates the complete two-round strategy, even though the success condition is checked only after the second round.
This quantity will serve as the common performance measure for the quantum and classical strategies considered in the following sections.

\section{A Quantum Circuit of Constant Depth} \label{sec:Quantum}

In this section, we construct a quantum circuit of constant depth for the 2R1D-CHSH problem.
The construction shows that all quantum gates can be chosen to act on adjacent registers in one dimension and that the resulting average success probability is equal to the quantum value of the CHSH game.

For clarity, we divide the circuit into two portions corresponding to the two rounds of the problem.
The first portion is executed during the first round and generates the output strings $\mathbf{s}$ and $\mathbf{t}$.
The second portion is executed during the second round and generates the output strings $\mathbf{a}$ and $\mathbf{b}$.
The complete circuit is shown in Fig.~\ref{fig:QuantumCircuit}.

\begin{figure*}[t]
\centering
\includegraphics[width = 0.96\textwidth]{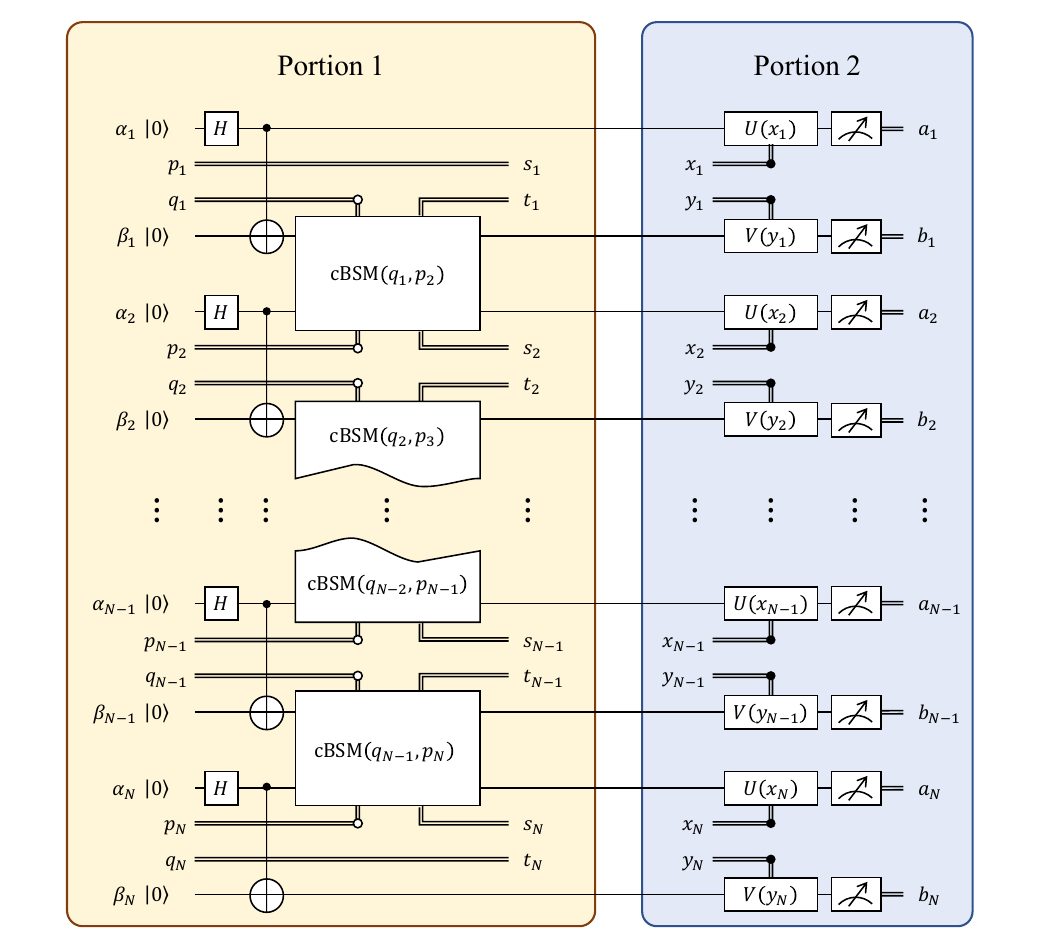}
\caption{Quantum circuit for the 2R1D-CHSH problem.
The first portion (left shaded area) receives the input pair $(\mathbf{p}[j], \mathbf{q}[k])$ of the first round, distributes entanglement between the selected positions, and generates the output pair $(\mathbf{s}, \mathbf{t})$.
The second portion (right shaded area) receives the input pair $(\mathbf{x}, \mathbf{y})$ of the second round and generates the output pair $(\mathbf{a}, \mathbf{b})$ by implementing the CHSH measurements.
An open control circle indicates that the corresponding operation is selected when the classical control bit is $0$.
In particular, $\operatorname{cBSM}(q_i, p_{i+1})$ acts when both $q_i$ and $p_{i+1}$ are $0$.
Each cBSM box includes the selected change to the Bell basis, the computational basis measurements, and the resets described in the text; in the inactive case, it leaves the quantum registers unchanged and reports the prescribed classical bits.
The gates $U(x_i)$ and $V(y_i)$ are defined in Eq.~(\ref{eq:MeasurementUnitaries}).}
\label{fig:QuantumCircuit}
\end{figure*}

The quantum circuit contains $2N$ quantum registers, denoted by $\alpha_i$ and $\beta_i$ for $i\in\{1, \ldots, N\}$.
Each register contains one qubit and is initially prepared in $\ket{0}$.
The classical input wires carrying $p_i$ and $q_i$ are placed locally beside the quantum registers that they control.
The wires carrying $x_i$ and $y_i$, which are supplied in the second round, are likewise placed beside $\alpha_i$ and $\beta_i$, respectively.
Including the classical wires available in the first round, the local ordering is given by the sequence
\begin{equation} \label{eq:QuantumRegisterOrdering}
\alpha_1, p_1, q_1, \beta_1, 
\alpha_2, p_2, q_2, \beta_2, 
\ldots, 
\alpha_N, p_N, q_N, \beta_N.
\end{equation}
The classical wires in Eq.~(\ref{eq:QuantumRegisterOrdering}) do not represent additional qubits.
Removing these wires gives the following induced ordering of the quantum registers:
\begin{equation} \label{eq:QuantumQubitOrdering}
\alpha_1, \beta_1, \alpha_2, \beta_2, \ldots, \alpha_N, \beta_N.
\end{equation}
In particular, both $\alpha_i$ and $\beta_i$, as well as $\beta_i$ and $\alpha_{i+1}$, are adjacent quantum registers.

We now describe the first portion of the circuit, which performs entanglement generation.
The first portion applies the Hadamard gate $H$ to every $\alpha_i$ and then applies a CNOT with control $\alpha_i$ and target $\beta_i$.
These gates prepare the Bell state defined in Eq.~(\ref{eq:Phi-plus}).
All Hadamard gates can be applied simultaneously, followed by the simultaneous application of all CNOT gates.

For every $i\in\{1, \ldots, N-1\}$, the first portion then applies the classically controlled instrument $\operatorname{cBSM}(q_i, p_{i+1})$ to the pair of qubits $\beta_i$ and $\alpha_{i+1}$.
Its action is determined by the classical bits $q_i$ and $p_{i+1}$.
If the following activation condition holds:
\begin{equation} \label{eq:cBSM_condition}
q_i = 0 = p_{i+1}, 
\end{equation}
the circuit applies a CNOT with control $\beta_i$ and target $\alpha_{i+1}$, applies $H$ to $\beta_i$, and measures both qubits in the computational basis.
This sequence implements the standard measurement in the Bell basis used in quantum teleportation~\cite{Bennett1993}.

The measurement results obtained from $\beta_i$ and $\alpha_{i+1}$ are reported as $t_i$ and $s_{i+1}$, respectively.
After these outcomes have been recorded, the measured qubits are reset to $\ket{0}$.
This reset does not affect the state of the unmeasured qubits because, conditioned on the recorded outcomes, the measured qubits are no longer entangled with the unmeasured registers.

If the condition in Eq.~(\ref{eq:cBSM_condition}) does not hold, the instrument leaves $\beta_i$ and $\alpha_{i+1}$ unchanged and reports the following deterministic outputs:
\begin{equation} \label{eq:cBSMInactiveOutputs}
t_i = q_i, \quad s_{i+1} = p_{i+1}.
\end{equation}
Since $\operatorname{cBSM}$ produces classical outcomes together with a conditional change of the quantum state, it is a quantum instrument in the standard sense~\cite{Davies1970, Ozawa1984}.

As shown in Fig.~\ref{fig:QuantumCircuit}, the boundary input bits $p_1$ and $q_N$ do not control any instance of $\operatorname{cBSM}$. They are copied directly to the corresponding outputs:
\begin{equation} \label{eq:BoundaryOutputs}
s_1 = p_1, \quad t_N = q_N.
\end{equation}
Thus, the first portion produces the output strings $\mathbf{s}$ and $\mathbf{t}$ required in the first round of the 2R1D-CHSH problem.

The first portion ends when $\mathbf{s}$ and $\mathbf{t}$ have been produced.
Only then is the input pair $(\mathbf{x}, \mathbf{y})$ of the second round supplied to the second portion.
This portion implements the binary observables defined in Eqs.~(\ref{eq:CHSH-Alice-observables}) and~(\ref{eq:CHSH-Bob-observables}).
More precisely, it performs the required changes of measurement basis by applying the following gates selected by the classical bits $x_i$ and $y_i$:
\begin{equation} \label{eq:MeasurementUnitaries}
\begin{aligned}
U(0)& = \mathbb{I}, &V(0)& = R_y(-\pi/4), \\
U(1)& = H, &V(1)& = R_y(\pi/4), 
\end{aligned}
\end{equation}
The rotation operator $R_y(\theta)$ is defined by $R_y(\theta)=e^{-i\theta Y/2}$, where $Y$ denotes the Pauli matrix
\begin{equation} \label{eq:Pauli-Y}
Y = \begin{pmatrix} 0 & -i\\ i & 0 \end{pmatrix}.
\end{equation}
The gates $U(x_i)$ and $V(y_i)$ act on $\alpha_i$ and $\beta_i$, respectively.
They are selected by classical inputs and are not controlled quantum gates.

Every quantum register is subsequently measured in the computational basis.
The results obtained from $\alpha_i$ and $\beta_i$ are reported as $a_i$ and $b_i$, respectively.
Since a computational basis measurement is a measurement of $Z$, the effective observables satisfy
\begin{equation} \label{eq:EffectiveObservables}
U(x_i)^\dagger ZU(x_i) = \mathcal{A}_{x_i}, \quad
V(y_i)^\dagger ZV(y_i) = \mathcal{B}_{y_i}, 
\end{equation}
where $\mathcal{A}_{x_i}$ and $\mathcal{B}_{y_i}$ are the observables defined in Eqs.~(\ref{eq:CHSH-Alice-observables}) and~(\ref{eq:CHSH-Bob-observables}).
Equation~(\ref{eq:EffectiveObservables}) therefore shows that the second portion implements the CHSH measurements.

We next specify the circuit depth. As usual, the depth of a quantum circuit is the number of sequential layers of operations, where operations with pairwise disjoint quantum supports may be performed within the same layer~\cite{Bravyi2020}.
For the present circuit, a local operation selected by a constant number of classical input bits is counted in the layer in which it acts.
Since the circuit contains intermediate measurements and resets, we additionally adopt the convention that a layer of measurements and a layer of resets each contribute one unit to the operational depth.
The initialization of the input qubits and the waiting time between the two rounds are not included.

The preparation of the initial Bell states requires two layers: one for the Hadamard gates on the registers $\alpha_i$ and one for the CNOT gates acting on $\alpha_i\beta_i$.
The controlled instruments $\operatorname{cBSM}(q_i, p_{i+1})$ require four further layers.
The selected CNOT gates are applied first, followed by the selected Hadamard gates, the measurements in the computational basis, and the resets of the measured qubits.
Within each of these layers, all operations act on disjoint quantum registers and can therefore be performed simultaneously.
Finally, the second portion requires one layer for the gates $U(x_i)$ and $V(y_i)$, followed by one layer for the measurements in the computational basis.
The operational depth of the quantum circuit $\mathcal Q_N$ is consequently
\begin{equation} \label{eq:OperationalDepth}
D_{\mathrm{op}}(\mathcal Q_N) = 2+4+2 = 8.
\end{equation}

If the measurement and reset within $\operatorname{cBSM}$ are regarded as a single quantum instrument, their two layers can be combined and the operational depth becomes $7$.
If only unitary layers are counted, the corresponding depth is $5$.
To avoid relying on either of these alternative conventions, we use the conservative bound $D_{\mathrm{op}}(\mathcal Q_N)\leq 8$ in Theorem~\ref{thm:ConstantDepthQuantumCircuit}.

The fan-in and fan-out are also bounded explicitly.
The gates $H$, $U(x_i)$, and $V(y_i)$, as well as each measurement and reset, act on one quantum register.
Every CNOT acts on two quantum registers.
Hence, the quantum fan-in of every elementary operation is at most $2$.
Each instance of $\operatorname{cBSM}$ depends on only two classical bits and acts on only two quantum registers.
Therefore, even when it is regarded as one hybrid operation, its total input arity is at most $4$.

A quantum output wire never branches and consequently has fan-out $1$.
Each classical input bit is used only within one local instance of $\operatorname{cBSM}$ or to select one basis rotation.
If $\operatorname{cBSM}$ is decomposed into the layers described above, each of its control bits is used by at most a constant number of local operations.
Similarly, each measurement outcome is used only as an output bit and, when needed, to control its local reset.
Thus, the classical fan-out is also bounded by a constant independent of $N$.
In particular, the circuit contains no gate whose fan-in or fan-out grows with $N$.

We now show that the fixed circuit $\mathcal Q_N$ works for every input in the support of the uniform distribution specified in Sec.~\ref{sec:2R1D}.
Fix $1\leq j<k\leq N$ and consider the input pair $(\mathbf{p}[j], \mathbf{q}[k])$ specified in Eq.~(\ref{eq:1Iset}).
For every $i\in\{j, \ldots, k-1\}$, one has $q_i = 0 = p_{i+1}$, so $\operatorname{cBSM}(q_i, p_{i+1})$ performs a measurement in the Bell basis on $\beta_i\alpha_{i+1}$. If $j>1$, the operation on $\beta_{j-1}\alpha_j$ is inactive because $p_j = 1$; if $j = 1$, $\alpha_j$ already lies at the left end of the register.
Similarly, if $k<N$, the operation on $\beta_k\alpha_{k+1}$ is inactive because $q_k = 1$; if $k = N$, $\beta_k$ lies at the right end.
Consequently, $\alpha_j$ and $\beta_k$ remain unmeasured and are isolated from the operations outside the interval from $j$ to $k$.
Measurements may also occur outside this interval, but they act on quantum registers separated from $\alpha_j$ and $\beta_k$ by the inactive operations or by the ends of the register.
They therefore have no effect on the state of the two selected registers.

To determine that state, consider two consecutive Bell states.
If the measurement of $\beta_i\alpha_{i+1}$ produces $(t_i, s_{i+1})$, the entanglement swapping identity gives the following transformation~\cite{Zukowski1993}:
\begin{equation} \label{eq:LocalEntanglementSwap}
\ket{\Phi_{0,0}}_{\alpha_i\beta_i}
\ket{\Phi_{0,0}}_{\alpha_{i+1}\beta_{i+1}}
\longmapsto
\ket{\Phi_{s_{i+1}, t_i}}_{\alpha_i\beta_{i+1}},
\end{equation}
conditioned on the recorded outcomes.
The framed Bell states are defined in Eq.~(\ref{eq:framed-Bell-state}).
Thus, $s_{i+1}$ contributes the $X$-frame bit $\sigma$, and $t_i$ contributes the $Z$-frame bit $\tau$.

The same identity applies when the Bell state on the left already carries a Pauli frame.
More precisely, for $j\leq i<k$, the measurement outcome $(t_i, s_{i+1})$ gives
\begin{equation}
\ket{\Phi_{\sigma,\tau}}_{\alpha_j\beta_i} \ket{\Phi_{0,0}}_{\alpha_{i+1}\beta_{i+1}}
\longmapsto
\ket{\Phi_{\sigma \oplus s_{i+1}, \tau \oplus t_i}}_{\alpha_j\beta_{i+1}}, 
\end{equation}
up to a physically irrelevant global phase.
This follows because the Pauli operator on the unmeasured register $\alpha_j$ commutes with the measurement on $\beta_i\alpha_{i+1}$, and the product of two Pauli operators adds their frame bits modulo two.

Consequently, conditioned on a fixed output $(\mathbf{s}, \mathbf{t})$ of the first round with nonzero probability, the normalized state of the selected registers is, up to a global phase, 
\begin{equation} \label{eq:RepeatedEntanglementSwap}
\begin{aligned}
\ket{\psi_{j,k}(\mathbf{s}, \mathbf{t})}
& = 
\ket{\Phi_{\sigma_{j,k}(\mathbf{s}), \tau_{j,k}(\mathbf{t})}}
_{\alpha_j\beta_k}
\\
& = 
\left(
Z^{\tau_{j,k}(\mathbf{t})}
X^{\sigma_{j,k}(\mathbf{s})}
 \otimes \mathbf 1
\right)
\ket{\Phi^+}_{\alpha_j\beta_k}.
\end{aligned}
\end{equation}
The two exponents are exactly the parities defined in Eq.~(\ref{eq:2R1D-frame-bits}).
Although this argument follows the measurements successively, they act on disjoint pairs of quantum registers and can be performed simultaneously.
Their sequential treatment is therefore only a proof device and does not increase the circuit depth.

The circuit does not compute either parity or apply the corresponding Pauli correction.
The parities appear only in the verification condition.
Moreover, the construction does not postselect on any measurement outcome: every transcript produced in the first round is retained.

For any fixed transcript of the first round with nonzero conditional probability, the gates $U(x_j)$ and $V(y_k)$ implement the CHSH measurements on the state in Eq.~(\ref{eq:RepeatedEntanglementSwap}).
The verification condition in Eq.~(\ref{eq:Relation}) is precisely the PF-CHSH winning condition for the frame $(\sigma_{j,k}(\mathbf{s}), \tau_{j,k}(\mathbf{t}))$.
Equation~(\ref{eq:PF-CHSH-pointwise-success}) therefore gives the conditional success probability
\begin{equation}
\frac{2+\sqrt{2}}{4} = \omega_{\mathrm{q}}(\mathrm{CHSH})
\end{equation}
for every question pair.
This value is independent of the selected positions, the transcript of the first round, and the remaining questions of the second round.
Its average over all these variables is therefore unchanged.

\begin{Thm} \label{thm:ConstantDepthQuantumCircuit}
For every $N\geq2$, there exists a fixed quantum circuit $\mathcal Q_N$ for the 2R1D-CHSH problem with operational depth at most $8$.
Every elementary quantum gate acts on at most two adjacent quantum registers, and every classically selected local operation depends on at most two classical input bits.
The quantum fan-out is $1$, and the classical fan-out is bounded by a constant.
The average success probability of this circuit is
\begin{equation} \label{eq:QuantumCircuitSuccess}
\overline{p}_{\mathrm{succ}}(\mathcal Q_N) = \omega_{\mathrm{q}}(\mathrm{CHSH}).
\end{equation}
\end{Thm}

\begin{proof}
For a fixed $N$, the locations and connections of all operations in Fig.~\ref{fig:QuantumCircuit} are independent of $j,k, \mathbf{x}, \mathbf{y}$.
These inputs determine only which local operations are selected.
Hence, Fig.~\ref{fig:QuantumCircuit} describes one fixed circuit $\mathcal Q_N$, rather than a different circuit for each input.

For every input pair $(\mathbf{p}[j], \mathbf{q}[k])$, the operations selected by Eq.~(\ref{eq:cBSM_condition}) produce a framed Bell state on $\alpha_j\beta_k$, as shown by Eq.~(\ref{eq:RepeatedEntanglementSwap}).
Its frame bits are the parity terms occurring in Eq.~(\ref{eq:Relation}).
By Eq.~(\ref{eq:EffectiveObservables}), the second portion implements the optimal CHSH observables.
The relation condition is therefore satisfied with probability $\omega_{\mathrm{q}}(\mathrm{CHSH})$ for every fixed input and every output of the first round with nonzero conditional probability.
Averaging over all of them gives Eq.~(\ref{eq:QuantumCircuitSuccess}).

The layer decomposition preceding Eq.~(\ref{eq:OperationalDepth}) shows that the operational depth is at most $8$, independently of $N$.
By Eq.~(\ref{eq:QuantumQubitOrdering}), every CNOT acts on adjacent quantum registers.
The remaining quantum operations act on individual registers. Thus, all quantum operations are geometrically local in one dimension.
Finally, every elementary quantum operation has fan-in at most $2$, each local classical selection depends on at most two input bits, and no quantum wire branches.
All fan-in and fan-out bounds are therefore independent of $N$.
\end{proof}

Theorem~\ref{thm:ConstantDepthQuantumCircuit} establishes the quantum part of the depth separation.
A crucial feature of the construction is that the circuit does not compute the parities associated with the Pauli frame; these parities appear only in the relation condition.
This allows the quantum depth to remain constant even when the selected positions are far apart.

The depth bound concerns the circuit that produces the outputs.
The classical computation used by the verifier to evaluate the relation, including the Pauli-frame parities, is not included in this bound.

\section{A Classical Upper Bound on the Average Success Probability} \label{sec:Classical}

In this section, we derive an upper bound on the average success probability of classical circuits for the 2R1D-CHSH problem.
This result provides the classical counterpart of Theorem~\ref{thm:ConstantDepthQuantumCircuit}.
The bound depends only on the depth and gate fan-in of the circuit portion executed after the questions of the second round become available.
The computation performed in the first round is unrestricted.

Let $\mathcal{C}$ be a probabilistic classical circuit that respects the two-round ordering specified in Sec.~\ref{sec:2R1D}.
We decompose $\mathcal{C}$ into its first portion $\mathcal{C}_1$ and second portion $\mathcal{C}_2$.
The circuit $\mathcal{C}_1$ produces the output strings $\mathbf{s}$ and $\mathbf{t}$ and may pass an arbitrary classical state to $\mathcal{C}_2$.
Upon receiving the second questions, $\mathcal{C}_2$ produces the output strings $\mathbf{a}$ and $\mathbf{b}$.
We impose no restriction on the depth, size, or fan-in of $\mathcal{C}_1$. Let $D$ be the depth of $\mathcal{C}_2$, and suppose that every gate in $\mathcal{C}_2$ has fan-in at most $F$, where $F\geq 2$.
The circuit $\mathcal{C}_2$ may have arbitrary size and unrestricted fan-out, and it need not satisfy any geometric locality condition.

For each $N$, the circuit graph of $\mathcal{C}_2$, including the locations of its question inputs and answer outputs, is fixed independently of the first input, transcript, and internal randomness.
The state of the first round may supply arbitrary classical data to $\mathcal{C}_2$, but it cannot change the circuit wiring.
Any data-dependent selection or routing must be implemented within $\mathcal{C}_2$ and is included in its depth.

All classical wires carry individual bits.
The fan-in bound counts every input bit on which a gate depends, including any bits used for control, selection, or routing.
The internal randomness of the circuit is independent of the verifier's randomness.
In particular, the questions of the second round are independent of the complete state retained from the first round and of all random bits used by the circuit.

We use the standard circuit-graph notion of a backward lightcone and the associated lightcone method for shallow-circuit lower bounds~\cite{Bravyi2020}.
For an output wire $o$ of $\mathcal{C}_2$, let $L_{\mathcal{C}_2}^{\leftarrow}(o)$ be its backward lightcone, namely, the set of input wires $u$ of $\mathcal{C}_2$ for which there exists a directed path from $u$ to $o$ in the circuit graph.
Thus, $L_{\mathcal{C}_2}^{\leftarrow}(o)$ contains precisely the input wires that may influence $o$ through $\mathcal{C}_2$.

Define the set of exceptional endpoint pairs for $\mathcal{C}_2$ by
\begin{widetext}
\begin{equation} \label{eq:exceptional-pair-set}
\mathcal{B}(\mathcal{C}_2) = \left\{ (j,k): 1\leq j<k\leq N, \quad y_k\in L_{\mathcal{C}_2}^{\leftarrow}(a_j) \text{ or } x_j\in L_{\mathcal{C}_2}^{\leftarrow}(b_k) \right\}.
\end{equation}
\end{widetext}
For every pair $(j,k)\notin\mathcal{B}(\mathcal{C}_2)$, the output $a_j$ cannot depend on the question wire $y_k$, and the output $b_k$ cannot depend on the question wire $x_j$.

\begin{Lem} \label{lem:exceptional-pair-probability}
The cardinality of the exceptional set is bounded in terms of $N$, $F$, and $D$ as follows:
\begin{equation} \label{eq:number-of-exceptional-pairs}
\left|\mathcal{B}(\mathcal{C}_2)\right| \leq 2NF^D.
\end{equation}
Consequently, the probability that a uniformly sampled endpoint pair is exceptional obeys
\begin{equation} \label{eq:exceptional-pair-probability}
\Pr\left[(J,K)\in\mathcal{B}(\mathcal{C}_2)\right] \leq \min\left\{ 1, \frac{4F^D}{N-1} \right\}.
\end{equation}
\end{Lem}

\begin{proof}
At each step of a backward traversal from an output wire $o$, each gate has at most $F$ input wires.
Thus, one additional layer increases the number of possible ancestors by at most a factor of $F$.
Since a directed path from an input to an output passes through at most $D$ gates, every output wire $o$ obeys the backward-lightcone bound
\begin{equation} \label{eq:backward-lightcone-size}
\left| L_{\mathcal{C}_2}^{\leftarrow}(o) \right| \leq F^D.
\end{equation}

Fix $j$. The backward lightcone of $a_j$ therefore contains at most $F^D$ input wires.
In particular, at most $F^D$ indices $k$ satisfy the following lightcone-inclusion condition:
\begin{equation}
y_k\in L_{\mathcal{C}_2}^{\leftarrow}(a_j).
\end{equation}
Summing over $j$ shows that the first condition in Eq.~(\ref{eq:exceptional-pair-set}) holds for at most $NF^D$ endpoint pairs.

The same argument, with $a_j$ and $y_k$ replaced by $b_k$ and $x_j$, shows that the second condition in Eq.~(\ref{eq:exceptional-pair-set}) holds for at most $NF^D$ endpoint pairs. A union bound therefore gives Eq.~(\ref{eq:number-of-exceptional-pairs}). Since the endpoint pair is uniform over the $\binom{N}{2}$ choices satisfying $j<k$, the fraction of exceptional pairs among all admissible endpoint pairs satisfies
\begin{equation}
\frac{\left|\mathcal{B}(\mathcal{C}_2)\right|}{\binom{N}{2}} \leq \frac{4F^D}{N-1}.
\end{equation}
Combining this estimate with the trivial upper bound of one proves Eq.~(\ref{eq:exceptional-pair-probability}).
\end{proof}

We next bound the success probability for an endpoint pair outside $\mathcal{B}(\mathcal{C}_2)$. The essential point is that $\mathcal{C}_1$ has been completed before the second questions are chosen. Hence the output strings $\mathbf{s}$ and $\mathbf{t}$, together with any classical state passed from $\mathcal{C}_1$ to $\mathcal{C}_2$, are independent of the second questions.

\begin{Lem} \label{lem:nonexceptional-pair-classical-bound}
For every pair $(j,k)\notin\mathcal{B}(\mathcal{C}_2)$, the success probability of $\mathcal{C}$, conditioned on $J = j$ and $K = k$ and averaged over the questions of the second round and the internal randomness of the circuit, is at most $3/4$.
\end{Lem}

\begin{proof}
Fix a pair $(j,k)\notin\mathcal{B}(\mathcal{C}_2)$.
Let $R$ denote all random bits used by the circuit, and let $M$ denote the complete classical state retained after the first round, including the output strings $\mathbf{S}$ and $\mathbf{T}$.
Define
\begin{equation}
\mathbf{X}_{-j}: = (X_i)_{i\ne j}, \quad
\mathbf{Y}_{-k}: = (Y_i)_{i\ne k}, 
\end{equation}
and collect the variables to be fixed into
\begin{equation}
\Lambda: = (R, M, \mathbf{X}_{-j}, \mathbf{Y}_{-k}).
\end{equation}
The questions of the second round are sampled independently and uniformly, independently of the circuit's randomness and its retained state.
Hence, for every value $\lambda$ with positive conditional probability, 
\begin{equation}
\Pr\left[ X_j = x, Y_k = y | J = j, K = k, \Lambda = \lambda \right] = \frac{1}{4}
\end{equation}
for all $x,y\in\{0,1\}$.

Fix such a value $\lambda$.
All random bits and all inputs of $\mathcal{C}_2$ other than $x_j$ and $y_k$ are now fixed.
Since $(j,k)\notin\mathcal{B}(\mathcal{C}_2)$, there is no directed path from $y_k$ to $a_j$, or from $x_j$ to $b_k$, in the circuit graph.
Consequently, there exist deterministic functions $A_\lambda, B_\lambda:\{0,1\}\to\{0,1\}$ such that
\begin{equation} \label{eq:local-functions-outside-exceptional-set}
a_j = A_\lambda(x_j), \quad
b_k = B_\lambda(y_k).
\end{equation}
Their dependence on the fixed indices $j$ and $k$ is suppressed.

The value of $\lambda$ also fixes the output strings $\mathbf{s}, \mathbf{t}$, and therefore fixes $\sigma_{j,k}(\mathbf{s})$ and $\tau_{j,k}(\mathbf{t})$.
Thus, Eq.~(\ref{eq:local-functions-outside-exceptional-set}) defines a deterministic classical local strategy for the corresponding PF-CHSH game with independent and uniform questions.
The classical value in Eq.~(\ref{eq:PF-CHSH-values}) therefore bounds its success probability by $3/4$.

This bound holds for every value of $\lambda$ with positive conditional probability.
Its average over $\Lambda$, conditioned on $J = j$ and $K = k$, is therefore also at most $3/4$.
\end{proof}

The preceding lemmas yield the following upper bound on the average success probability of classical circuits for the 2R1D-CHSH problem.

\begin{Thm} \label{thm:classical-average-success-bound}
Let $N\geq2$, and let $\mathcal{C}$ be a probabilistic classical circuit for the 2R1D-CHSH problem in the circuit model specified above.
In particular, the wiring of its second portion $\mathcal{C}_2$, including the locations of its question inputs and answer outputs, is fixed independently of the input and transcript of the first round.
Suppose that $\mathcal{C}_2$ has depth $D$ and that every gate has fan-in at most $F$, where $F\geq2$.
Under the input distribution specified in Sec.~\ref{sec:2R1D}, its average success probability satisfies
\begin{equation} \label{eq:classical-average-success-bound}
\overline{p}_{\mathrm{succ}}(\mathcal{C}) \leq \frac{3}{4}+ \min\left\{ \frac{1}{4}, \frac{F^D}{N-1} \right\}.
\end{equation}
No restriction is imposed on the depth, size, or fan-in of $\mathcal{C}_1$, or on the size, fan-out, or geometric locality of $\mathcal{C}_2$.
\end{Thm}

\begin{proof}
Let $q$ denote the probability that the sampled endpoint pair is exceptional; explicitly, 
\begin{equation}
q = \Pr\left[ (J,K)\in\mathcal{B}(\mathcal{C}_2) \right].
\end{equation}
The endpoint-pair distribution can be partitioned into the event $(J,K)\in\mathcal{B}(\mathcal{C}_2)$, which has probability $q$, and its complementary event, which has probability $1-q$.
On the complementary event, Lemma~\ref{lem:nonexceptional-pair-classical-bound} bounds the conditional success probability by $3/4$.
On the exceptional event, the success probability is bounded by the trivial upper bound of one.
Averaging these two conditional bounds with weights $1-q$ and $q$, respectively, gives
\begin{equation} \label{eq:success-bound-in-terms-of-exceptional-pairs}
\overline{p}_{\mathrm{succ}}(\mathcal{C}) \leq \frac{3}{4}(1-q)+q = \frac{3}{4}+\frac{q}{4}.
\end{equation}
Substituting the estimate from Lemma~\ref{lem:exceptional-pair-probability} into Eq.~(\ref{eq:success-bound-in-terms-of-exceptional-pairs}) gives Eq.~(\ref{eq:classical-average-success-bound}).
\end{proof}

Theorem~\ref{thm:classical-average-success-bound} depends only on the depth and fan-in of $\mathcal{C}_2$.
In particular, arbitrary classical computation performed by $\mathcal{C}_1$ before the second questions are chosen cannot overcome the restriction imposed by a shallow $\mathcal{C}_2$.

\section{The Depth-Separation Theorem} \label{sec:DepthSeparation}

In this section, we derive a lower bound on classical circuit depth from the upper bound established in Theorem~\ref{thm:classical-average-success-bound} and compare it with the constant-depth quantum circuit constructed in Sec.~\ref{sec:Quantum}.
We first derive the relevant depth condition and then combine the quantum and classical results into a separation theorem.

The quantum circuit $\mathcal Q_N$ constructed in Sec.~\ref{sec:Quantum} has operational depth at most eight, independently of $N$, and achieves the average success probability stated in Eq.~(\ref{eq:QuantumCircuitSuccess}).
To compare this circuit with a classical circuit, we allow the average success probability of the classical circuit to be up to an additive amount $\delta$ below that of the quantum circuit.
We require this tolerance to be smaller than the gap between the quantum and classical values of the CHSH game.
This requirement is expressed by the following condition:
\begin{equation} \label{eq:admissible-additive-error}
0 \leq \delta < \frac{\sqrt{2}-1}{4}.
\end{equation}
After allowing this tolerance, a positive advantage over the classical CHSH value $3/4$ still remains.
We denote this remaining advantage by $\epsilon_{\delta}$, which is given explicitly by
\begin{equation} \label{eq:classical-advantage-epsilon}
\epsilon_{\delta} = \frac{\sqrt{2}-1}{4}-\delta.
\end{equation}
The restriction in Eq.~(\ref{eq:admissible-additive-error}), together with the definition in Eq.~(\ref{eq:classical-advantage-epsilon}), guarantees that $\epsilon_{\delta}>0$.

Now consider a probabilistic classical circuit $\mathcal{C}$ whose second portion $\mathcal{C}_2$ has depth $D$ and maximum gate fan-in at most $F$.
Suppose that the average success probability of $\mathcal{C}$ is no more than $\delta$ below that of $\mathcal Q_N$.
This requirement can be written as
\begin{equation} \label{eq:classical-quantum-approximation}
\overline{p}_{\mathrm{succ}}(\mathcal{C}) \geq \overline{p}_{\mathrm{succ}}(\mathcal Q_N)-\delta.\end{equation}
This requirement concerns only the average success probability; it does not require the classical circuit to approximate the full output distribution of $\mathcal Q_N$.
By Eq.~(\ref{eq:QuantumCircuitSuccess}), this condition requires $\mathcal{C}$ to achieve an average success probability of at least $3/4+\epsilon_{\delta}$.
In other words, the excess term in the classical upper bound of Theorem~\ref{thm:classical-average-success-bound} must be at least $\epsilon_{\delta}$.
The fan-in and depth of $\mathcal{C}_2$ must therefore satisfy the following necessary condition:
\begin{equation} \label{eq:fan-in-depth-requirement}
F^D \geq \epsilon_{\delta}(N-1).
\end{equation}
Equation~(\ref{eq:fan-in-depth-requirement}) implies $D\geq\log_F[\epsilon_\delta(N-1)]$.
For fixed $F$ and fixed $\delta$, this lower bound grows logarithmically with $N$.
Since circuit depth is a nonnegative integer, the resulting lower bound must be rounded up and bounded below by zero.
Combining this observation with the quantum construction gives the main theorem.

\begin{Thm}[Depth separation] \label{thm:depth-separation}
Fix a maximum gate fan-in $F\geq 2$, and let $\delta$ satisfy Eq.~(\ref{eq:admissible-additive-error}).
For every $N\geq 2$, the 2R1D-CHSH problem of size $N$ admits a quantum circuit $\mathcal Q_N$ with operational depth $D_{\mathrm{op}}(\mathcal Q_N)\leq 8$ and average success probability given by Eq.~(\ref{eq:QuantumCircuitSuccess}).

In contrast, let $\mathcal{C}$ be any probabilistic classical circuit for the 2R1D-CHSH problem of size $N$ in the fixed-wiring circuit model of Sec.~\ref{sec:Classical}.
Suppose that its second portion $\mathcal{C}_2$ has depth $D$ and maximum gate fan-in at most $F$, and that $\mathcal{C}$ satisfies the approximation requirement in Eq.~(\ref{eq:classical-quantum-approximation}).
Then the depth of $\mathcal{C}_2$ obeys the following lower bound:
\begin{equation} \label{eq:classical-depth-lower-bound}
D \geq \max\left\{ 0, \left\lceil \log_{F}\left[ (N-1) \left( \frac{\sqrt{2}-1}{4}-\delta \right) \right] \right\rceil \right\}.
\end{equation}
Consequently, for every fixed $F$ and every fixed $\delta<(\sqrt{2}-1)/4$, any such classical circuit family requires depth $D = \Omega(\log N)$ in its second portion, whereas the quantum circuit family has operational depth at most eight.
\end{Thm}

The quantum part of Theorem~\ref{thm:depth-separation} shows that, for every problem size $N$, the 2R1D-CHSH problem admits a quantum circuit of depth at most eight whose average success probability is exactly the quantum value of the CHSH game.
Thus, increasing $N$ neither increases the required quantum depth nor decreases the achievable average success probability.

The classical part must be interpreted relative to the classical CHSH value $3/4$.
The theorem does not rule out a constant-depth classical circuit family that attains this baseline value. Rather, it shows that no bounded-fan-in, fixed-wiring classical circuit family of constant depth can maintain an $N$-independent positive advantage above $3/4$.
In particular, such a classical circuit family cannot match the quantum CHSH value, or remain within any fixed additive error $\delta<(\sqrt{2}-1)/4$ of that value, uniformly as $N$ increases.

The logarithmic dependence on $N$ is optimal up to constant factors.
To see this, consider a classical strategy whose first portion outputs $\mathbf{s} = \mathbf{t} = \mathbf0$ and retains $\mathbf{p}[j]$.
In the second portion, for $1\leq i\leq N$, define
\begin{equation}
\xi: = \bigvee_{r = 1}^{N}(p_r\land x_r), \quad
a_i: = 0, \quad
b_i: = \xi\land y_i, 
\end{equation}
where $\land$ and $\bigvee$ denote Boolean AND and OR, respectively.
Since $\mathbf{p}[j]$ has its unique nonzero entry at $j$, one has $\xi = x_j$. Both Pauli-frame bits vanish, and therefore
\begin{equation}
a_j \oplus b_k = x_jy_k.
\end{equation}
The strategy succeeds with probability one for every allowed input.

A layer of binary AND gates computes the terms $p_r\land x_r$.
A balanced tree of binary OR gates computes $\xi$ in depth $\lceil\log_2N\rceil$, and a final layer computes all outputs $b_i$.
Thus, the second portion has size $O(N)$, fan-in at most two, and depth at most $\lceil\log_2N\rceil+2$.
Its wiring is fixed independently of the selected positions.
Together with Theorem~\ref{thm:classical-average-success-bound}, this establishes the optimal $\Theta(\log N)$ scaling of the classical second-round depth required to achieve success probability at least $3/4 + \varepsilon$, for every fixed $0 < \varepsilon \le 1/4$.

The class $\mathrm{AC}^{0}$ consists of decision problems solvable by Boolean circuits of constant depth and size polynomial in the input length, using AND and OR gates of unbounded fan-in and NOT gates~\cite{Watts2019}.
For relation problems, the term $\mathrm{AC}^{0}$ circuit refers to the same circuit restrictions with multiple output bits.

If unbounded fan-in OR gates are allowed, the OR tree in the preceding strategy can be replaced by one gate.
The second portion then has depth three and size $O(N)$, while the first portion only retains $\mathbf p[j]$ and outputs constant frame bits.
Thus, the strategy admits an implementation by $\mathrm{AC}^{0}$ circuits that respects the ordering of the two rounds.
The present relation therefore does not yield a separation against this circuit class.

\section{Further Remarks} \label{sec:Remarks}

In this section, we provide two further remarks on the structure and practical implications of the 2R1D-CHSH problem.
We first explain why the separation of the problem into two rounds is essential for obtaining the depth separation by contrasting it with the corresponding one round formulation.
We then examine the finite size qubit requirements for demonstrating quantum advantage and compare them with those of the 1D Magic Square and Magic Pentagram problems.

\subsection{Why Two Rounds Are Required} \label{sec:Why}

The two-round structure of the 2R1D-CHSH problem is essential rather than a technical convenience.
To see this, consider the one-round one-dimensional CHSH (1R1D-CHSH) problem obtained by removing the temporal separation between the two rounds.
For a fixed $N\geq 2$, the input consists of the strings $(\mathbf{p}[j], \mathbf{q}[k])$ that encode the positions and the strings $(\mathbf{x}, \mathbf{y})$ that specify the CHSH questions, all of which are supplied simultaneously.
We write this combined input as
\begin{equation}
\left((\mathbf{p}[j], \mathbf{q}[k]), (\mathbf{x}, \mathbf{y})\right).
\end{equation}
The circuit produces, in the same round, an output consisting of the Pauli-frame strings $(\mathbf{s}, \mathbf{t})$ and the CHSH-answer strings $(\mathbf{a}, \mathbf{b})$.
We write this combined output as
\begin{equation}
\left((\mathbf{s}, \mathbf{t}), (\mathbf{a}, \mathbf{b})\right).
\end{equation}
The input and output alphabets, their one-dimensional arrangement, the input distribution, the Pauli-frame bits in Eq.~(\ref{eq:2R1D-frame-bits}), and the validity condition in Eq.~(\ref{eq:Relation}) are otherwise unchanged.
Thus, the only difference from the 2R1D-CHSH problem is the causal order of the inputs and outputs.
In the one-round formulation, the circuit may choose the Pauli-frame outputs $(\mathbf{s}, \mathbf{t})$ and the CHSH-answer outputs $(\mathbf{a}, \mathbf{b})$ jointly after observing both $(\mathbf{p}[j], \mathbf{q}[k])$ and $(\mathbf{x}, \mathbf{y})$.

The simultaneous availability of these inputs permits a simple deterministic classical strategy.
At every site $i\in\{1, \ldots, N\}$, the circuit assigns its Pauli-frame and answer output bits according to the following local rules:
\begin{equation}
s_i = q_i y_i, \quad t_i = 0, \quad a_i = 0, \quad b_i = y_i.
\end{equation}
Because $\mathbf{q}[k]$ has its unique nonzero entry at position $k$, substituting these outputs into the definitions in Eq.~(\ref{eq:2R1D-frame-bits}) yields the following values of the two relevant Pauli-frame bits:
\begin{equation}
\sigma_{j,k}(\mathbf{s}) = \bigoplus_{i = j+1}^{k}q_i y_i = y_k, \quad
 \tau_{j,k}(\mathbf{t}) = 0.
\end{equation}
Using these frame-bit values, the expression on the right-hand side of the validity condition in Eq.~(\ref{eq:Relation}) simplifies as follows:
\begin{align}
&x_jy_k \oplus (1-x_j)\sigma_{j,k}(\mathbf{s}) \oplus x_j\tau_{j,k}(\mathbf{t}) \nonumber \\
& = x_jy_k \oplus (1-x_j)y_k \\
& = y_k \\
& = a_j \oplus b_k.
\end{align}
Hence the validity condition holds for every $1\leq j<k\leq N$ and every input $(\mathbf{x}, \mathbf{y})$.

This strategy can be implemented by a geometrically local classical circuit of depth one.
At each site, one AND gate computes $s_i = q_i y_i$, a wire copies $y_i$ to $b_i$, and the outputs $t_i$ and $a_i$ are set to zero.
All AND gates act independently and in parallel.
Neither the circuit depth nor the interaction range depends on $N$.
The preceding construction establishes the following proposition.

\begin{Prop} \label{prop:1R1D-classical-strategy}
For every $N\geq 2$, the 1R1D-CHSH problem admits a geometrically local classical circuit of depth one whose success probability is one.
\end{Prop}

The same input-output rule also admits a quantum implementation of constant depth.
At each site, a reversible AND operation writes $q_i y_i$ into a register initialized in $\ket{0}$, and a CNOT copies the computational basis value of $y_i$ into the register for $b_i$.
The registers for $a_i$ and $t_i$ remain in $\ket{0}$.
These local operations require only a constant number of gates acting on one or two qubits and can be performed in parallel across sites.
Since the inputs are classical, the same rule can also be implemented using local classically selected operations, without entanglement as a computational resource.

The perfect strategy above is possible because the one-round circuit can adapt the Pauli frame to the CHSH questions.
In particular, the choice $s_i = q_i y_i$ gives $\sigma_{j,k}(\mathbf{s}) = y_k$, reducing the verification condition to $a_j \oplus b_k = y_k$, which is satisfied by a local response.
The two-round formulation rules out precisely this dependence: the outputs $(\mathbf{s}, \mathbf{t})$, and hence the Pauli frame, must be fixed before the independently sampled input $(\mathbf{x}, \mathbf{y})$ is supplied.
The frame bits therefore cannot be chosen as functions of $x_j$ or $y_k$, which allows the classical value of the fixed-frame PF-CHSH game to enter the lightcone argument.
Since the 1R1D-CHSH problem instead has a constant-depth classical strategy with success probability one, no quantum-classical depth separation based on the success probability can be obtained from its one-round formulation.

\subsection{Sufficient Data-Qubit Counts from Analytical Bounds} \label{sec:Qubit}

We compare sufficient qubit counts for the 2R1D-CHSH problem, the 1D Magic Square problem~\cite{Bravyi2020}, and the Magic Pentagram problem~\cite{Han2022}.
Throughout this comparison, the maximum classical gate fan-in is $F = 2$.
The qubit counts refer to the data registers of the ideal quantum circuits and exclude any additional resources for control, routing, or error correction.

The benchmark is an additive advantage of at least $0.01$ between the success probability of the specified quantum circuit and an upper bound for the competing classical circuits.
A size certified in this way is sufficient, but need not be the smallest size at which quantum advantage occurs.

For consistency, we apply the same averaging argument to all three classical estimates.
If the success probability conditioned on an event $E$ is at most $c$, and $q = \Pr[E^c]$, then
\begin{equation} \label{eq:comparison-averaging}
p_{\mathrm{succ}} \leq c(1-q)+q = c+(1-c)q.
\end{equation}
For the 1D Magic Square problem, Lemma~7 in Sec.~I.C.2 of the Supplementary Information of Ref.~\cite{Bravyi2020} gives $q \leq 80F^{2D}/N$, while Lemmas~3 and~8 give the conditional success bound $c = 8/9$.
These results are combined in the proof of Theorem~4 of that Supplementary Information.
For the Magic Pentagram problem, Proposition~13 and Lemma~14 in Appendix~E of Ref.~\cite{Han2022} give $q \leq 216F^{2D}/N$ and $c = 19/20$, respectively.
They are combined in the proof of Theorem~10 in Appendix~F of that reference.
We use $N$ for the size parameter $n$ in both references, and $F$ for their respective fan-in parameters $K$ and $B$.
For these two problems, $D$ denotes the depth of the entire competing classical circuit.

Let $p_{\mathrm{MS}}$ and $p_{\mathrm{MP}}$ denote the average classical success probabilities under the uniform distributions on the restricted instance sets used in these results.
Equation~(\ref{eq:comparison-averaging}) gives
\begin{equation} \label{eq:magic-square-refined-classical-bound}
p_{\mathrm{MS}} \leq \frac{8}{9}+\frac{80F^{2D}}{9N},
\end{equation}
and
\begin{equation} \label{eq:magic-pentagram-classical-bound}
p_{\mathrm{MP}} \leq \frac{19}{20}+\frac{216F^{2D}}{20N}.
\end{equation}
The final averaging steps in the cited proofs use the weaker estimate $p_{\mathrm{succ}} \leq c+q$.
The bounds above retain the factor $1-c$ without changing the exceptional events or the conditional success bounds.
Appendix~\ref{app:comparison-bounds} specifies these events and provides the derivation.

For the 2R1D-CHSH problem, we allow the second portion of the classical circuit to have depth $D = 8$, while the total operational depth of the quantum circuit is at most eight.
This comparison is conservative because the classical computation in the first round remains unrestricted.
Theorem~\ref{thm:classical-average-success-bound} then certifies an advantage of at least $0.01$ when
\begin{equation}
N \geq 1+\left\lceil \frac{2^8}{(\sqrt{2}-1)/4-0.01} \right\rceil = 2738.
\end{equation}
The corresponding quantum circuit uses $2N = 5476$ data qubits.

For the 1D Magic Square problem, the circuit in Fig.~2 of Ref.~\cite{Bravyi2020} uses $4N$ data qubits.
The discussion following Theorem~2 in Sec.~I.B of its Supplementary Information gives a depth-five implementation on the restricted input set using classically selected Clifford gates acting on at most two qubits.
This depth counts the unitary gates and excludes the final computational basis measurements.
Under the operational depth convention specified in Sec.~\ref{sec:Quantum}, the final measurement layer gives an operational depth bound of six.
We therefore set the competing classical depth to $D = 6$ in Eq.~(\ref{eq:magic-square-refined-classical-bound}).
Since the ideal quantum circuit succeeds with probability one, the benchmark condition becomes
\begin{equation} \label{eq:magic-square-qubit-benchmark}
N \geq \left\lceil \frac{80\cdot2^{12}}{9(1-0.01-8/9)} \right\rceil = 360\,088.
\end{equation}
The corresponding count is $4N = 1\,440\,352$ data qubits.

The Magic Pentagram construction uses $6N$ data qubits and also succeeds with probability one~\cite{Han2022}.
Based on the construction in that reference, we calculate an operational depth bound of nine on the restricted instance set used in this comparison.
The Bell states can be prepared in two layers.
The local basis changes require at most four layers, with the most demanding case consisting of three sequential two-qubit gates followed by one layer of Hadamard gates.
On the restricted instance set, the inverse basis changes within the entanglement swapping operations are identities, so these operations reduce to Bell basis changes on disjoint pairs and require two layers.
The final computational basis measurements contribute one additional layer, giving an operational depth bound of $2+4+2+1=9$.
We therefore set the competing classical depth to $D = 9$ in Eq.~(\ref{eq:magic-pentagram-classical-bound}).
The sufficient condition for an advantage of at least $0.01$ becomes
\begin{equation} \label{eq:magic-pentagram-qubit-benchmark}
N \geq \frac{216\cdot2^{18}}{20(1-0.01-19/20)} = 70\,778\,880.
\end{equation}
The corresponding count is $6N = 424\,673\,280$ data qubits.

\begin{table*}[t]
\caption{Sufficient qubit counts for an additive advantage of $0.01$ under the comparisons described in the text.
The maximum classical gate fan-in is $F = 2$.
The quantum depth bounds include the final measurement layer.
For Magic Square and Magic Pentagram, these depth bounds apply to the respective restricted instance sets used in the comparison.
For 2R1D-CHSH, $D = 8$ is allowed for the second portion of the classical circuit alone.
For Magic Square and Magic Pentagram, the benchmarks use $D = 6$ and $D = 9$, respectively, for the entire classical circuit.
The classical bounds for these two problems retain the factors $1-c$ in the averaging argument described in the text.
All qubit counts refer to ideal data registers.
}
\label{tab:qubit-comparison}
\begin{ruledtabular}
\begin{tabular}{lcccc}
Problem & Quantum depth & Data qubits & Classical upper bound & $0.01$ gap benchmark
\\
\hline
2R1D-CHSH
& $\leq8$ & $2N$ & $\displaystyle\frac{3}{4}+\frac{F^D}{N-1}$ & $N = 2738;\ 5476$ qubits \\
1D Magic Square~\cite{Bravyi2020}
& $\leq6$ & $4N$ & $\displaystyle\frac{8}{9}+\frac{80F^{2D}}{9N}$ & $N = 360\,088;\ 1\,440\,352$ qubits \\
Magic Pentagram~\cite{Han2022}
& $\leq9$ & $6N$ & $\displaystyle\frac{19}{20}+\frac{216F^{2D}}{20N}$ & $N = 70\,778\,880$; $424\,673\,280$ qubits
\end{tabular}
\end{ruledtabular}
\end{table*}

Table~\ref{tab:qubit-comparison} shows that these bounds certify the prescribed advantage with fewer data qubits for the 2R1D-CHSH construction than for either the stated Magic Square or Magic Pentagram benchmark.
The comparison nevertheless involves different relation problems, different numbers of rounds, and the gate implementations specified for each construction.
It is not a comparison of optimal circuit realizations or minimum experimental resources.

In particular, the present analysis assumes ideal operations.
Errors in the recorded outcomes of entanglement swapping can accumulate in the Pauli-frame parities even when the circuit depth is constant.
The quoted qubit counts therefore do not establish experimental feasibility or tolerance to a fixed local noise rate.
Such conclusions require a separate analysis of noise, control, measurement, and verification resources.

\section{Conclusion} \label{sec:Conclusion}

The 2R1D-CHSH problem exhibits an asymptotically tight difference between the depth requirements of quantum and classical circuits.
The quantum construction uses $2N$ qubits and maintains the CHSH quantum success probability at operational depth at most eight, independently of $N$.
In the classical model considered in this work, every fixed target success probability above $3/4$ requires logarithmic depth after the questions of the second round are supplied.
This lower bound is optimal in its dependence on $N$, since an explicit classical circuit of logarithmic depth solves the relation with certainty.
The advantage therefore concerns the depth needed to respond to the newly supplied questions with the prescribed success probability, rather than the absence of a classical solution to the relation.

Related constructions based on the Magic Square and Magic Pentagram games use quantum strategies that succeed with probability one~\cite{Bravyi2020,Han2022}.
By contrast, our circuit attains the optimal quantum value of CHSH, $(2+\sqrt2)/4<1$, while still yielding an unconditional separation in circuit depth.
Thus, quantum pseudo-telepathy~\cite{Brassard2005} is not necessary for this form of advantage.
The two rounds fix the Pauli-frame data before the questions are supplied, and the lightcone bound converts the gap between the classical and quantum CHSH values into a lower bound on classical circuit depth.

The separation also holds under an asymmetric comparison that favors the classical model.
The quantum circuit is restricted to adjacent gates in one dimension, whereas the classical bound imposes no geometric locality restriction and permits arbitrary circuit size and fan-out within the model of Sec.~\ref{sec:Classical}.
Only the depth and gate fan-in of the second portion enter the bound; the computation in the first round is unrestricted.
Moreover, the correction term $F^D/(N-1)$ makes the dependence on fan-in, depth, problem size, and success probability explicit.
These features identify limited information propagation as the obstruction faced by classical circuits, while one Bell pair per site enables the quantum circuit to produce the required correlations in constant depth.

Complementary separations for different relation or interactive problems have been established against $\mathrm{AC}^{0}$ and, for every fixed prime $p$, $\mathrm{AC}^{0}[p]$~\cite{Watts2019,Grier2020}.
The latter class extends $\mathrm{AC}^{0}$ by allowing $\operatorname{MOD}_p$ gates of unbounded fan-in.
Such a gate outputs $1$ exactly when the number of input bits equal to $1$ is divisible by $p$.

Several questions remain open.
The most immediate is whether a version of the 2R1D-CHSH problem that tolerates noise can preserve an advantage over $3/4$ independent of $N$ under a constant local noise rate, while retaining geometric locality in one dimension and linear qubit overhead.
Previous constructions tolerate noise through the rigidity of the Magic Square game or within interactive settings based on shallow Clifford circuits~\cite{Bravyi2020,Grier2021}.
It remains unclear whether the quantum value of CHSH, which is strictly smaller than one, permits an analogous construction.
The perfect $\mathrm{AC}^{0}$ strategy described in Sec.~\ref{sec:DepthSeparation} also shows that a separation against classical circuits with unbounded fan-in would require a different relation or additional restrictions on the computational model.

Another question is whether a different encoding can yield a CHSH construction with a single round, despite the obstruction identified in Sec.~\ref{sec:Why}.
More broadly, the theory of nonlocal games provides many examples in which entanglement produces a strict, but not necessarily perfect, advantage over classical strategies~\cite{Cleve2004,Brunner2014}.
One may therefore ask which nonlocal games with quantum value strictly smaller than one admit a construction of the type developed in this work, with constant circuit depth and a verification condition that accounts for the Pauli frame.
Such a characterization could lead to a general framework in which the separation in circuit depth is controlled directly by the gap between the classical and quantum values of the underlying game.

\section*{Data Availability}
The analytical results and the information required to reproduce the resource estimates are provided in this article.

\section*{AI Statement}

The authors used ChatGPT (GPT-5.6 Sol) for language revision and discussions of proofs and resource estimates, directing the interactions and critically evaluating and revising selected suggestions before incorporating them into the manuscript.
All mathematical arguments and calculations incorporated from these interactions were independently verified by the authors.
The authors take full responsibility for the accuracy and integrity of the final manuscript.

\begin{acknowledgments}
This research was supported by Basic Science Research Program through the National Research Foundation of Korea (NRF) funded by the Ministry of Education (Grant No. RS-2023-00243988).
S.L. acknowledges support from the NRF grants funded by the Ministry of Science and ICT (MSIT) (No. RS-2024-00432214 and No. RS-2022-NR068791) and Creation of the Quantum Information Science R\&D Ecosystem (No. RS-2023-NR068116) through the NRF funded by the MSIT as well as the Institute of Information \& Communications Technology Planning \& Evaluation (IITP) grant funded by the MSIT (No. RS-2025-02304540).
\end{acknowledgments}

\appendix

\section{CHSH Values and Pauli-Frame Correlations} \label{app:CHSH-details}

We provide the derivations of the game values and the conditional winning probability stated in Sec.~\ref{sec:CHSHg}.

\subsection{Classical values and output relabeling}
A deterministic classical CHSH strategy is specified by four bits $a_0, a_1, b_0, b_1$. A strategy that wins on all
question pairs would satisfy
\begin{align}
a_0 \oplus b_0& = 0, & a_0 \oplus b_1& = 0, \\
a_1 \oplus b_0& = 0, & a_1 \oplus b_1& = 1.
\end{align}
The XOR of these equations gives $0 = 1$.
Thus, every deterministic strategy loses on at least one of the four equally likely question pairs.
Shared randomness cannot improve the average because it gives a convex combination of deterministic strategies.
The strategy $a = b = 0$ attains $3/4$, which proves the classical value in Eq.~(\ref{eq:CHSH-values})~\cite{CHSH1969, Cleve2004}.

For a fixed Pauli frame, the output relabeling
\begin{equation}
a' = a \oplus (1-x)\sigma \oplus x\tau
\end{equation}
is its own inverse and depends only on Alice's question and the frame.
It therefore gives a bijection between strategies for CHSH and PF-CHSH without changing their winning probabilities.
This holds for both classical and quantum strategies and proves the equalities of the game values in Eq.~(\ref{eq:PF-CHSH-values}).
In particular, the classical strategy $a = (1-x)\sigma \oplus x\tau$, $b = 0$ attains $3/4$.

\subsection{Quantum correlations}
In the computational basis, the Pauli matrices are
\begin{equation}
Z = \begin{pmatrix}1&0\\0&-1\end{pmatrix}, \quad
X = \begin{pmatrix}0&1\\1&0\end{pmatrix}.
\end{equation}
The Bell state $\ket{\Phi^+}$ obeys
\begin{align}
\bra{\Phi^+}Z \otimes Z\ket{\Phi^+} & = \bra{\Phi^+}X \otimes X\ket{\Phi^+} = 1, \\
\bra{\Phi^+}Z \otimes X\ket{\Phi^+} & = \bra{\Phi^+}X \otimes Z\ket{\Phi^+} = 0.
\end{align}
Consequently, the observables in Eqs.~(\ref{eq:CHSH-Alice-observables}) and~(\ref{eq:CHSH-Bob-observables}) give
\begin{equation}
\bra{\Phi^+}\mathcal{A}_x \otimes \mathcal{B}_y\ket{\Phi^+} = \frac{(-1)^{xy}}{\sqrt{2}}.
\end{equation}

Let $P_{\sigma,\tau} = Z^\tau X^\sigma$.
The Pauli commutation relations imply
\begin{align}
P_{\sigma,\tau}^{\dagger}ZP_{\sigma,\tau} & = (-1)^\sigma Z, \\
P_{\sigma,\tau}^{\dagger}XP_{\sigma,\tau} & = (-1)^\tau X.
\end{align}
Thus, the correlations of the framed Bell state are
\begin{align}
E_{xy}(\sigma,\tau)
&: = \bra{\Phi_{\sigma,\tau}} \mathcal{A}_x \otimes \mathcal{B}_y \ket{\Phi_{\sigma,\tau}} \\
& = \frac{(-1)^{xy \oplus (1-x)\sigma \oplus x\tau}}{\sqrt{2}}.
\end{align}

The product of the two measurement eigenvalues is $(-1)^{a \oplus b}$.
For any required parity $r\in\{0,1\}$, the corresponding probability is therefore
\begin{equation}
\Pr[a \oplus b = r| x,y] = \frac{1+(-1)^rE_{xy}(\sigma,\tau)}{2}.
\end{equation}
For $r = xy \oplus (1-x)\sigma \oplus x\tau$, this expression equals $(2+\sqrt{2})/4$, proving Eq.~(\ref{eq:PF-CHSH-pointwise-success}) for every frame and question pair.
The case $\sigma = \tau = 0$ gives the same success probability for CHSH.
Tsirelson's bound shows that no quantum CHSH strategy can exceed this average winning probability~\cite{Tsirelson1980}.
The output relabeling above establishes the corresponding optimality for PF-CHSH.

\section{Classical Bounds Used in the Qubit Comparison} \label{app:comparison-bounds}

We derive the estimates used in Sec.~\ref{sec:Qubit} from the lightcone results in Refs.~\cite{Bravyi2020, Han2022}.
We use $N$ for the size parameter $n$ in both references, and $F$ for their fan-in parameters $K$ and $B$, respectively.
The depth $D$ in this appendix is the depth of the entire classical circuit for the Magic Square or Magic Pentagram problem.
We first fix all internal random bits of the circuit.
The estimates below hold for every such choice and therefore extend to probabilistic circuits by averaging.

\subsection{Input distributions and lightcone events}

For the Magic Square problem, the input is uniform on the set $S$ defined by Eqs.~(7) and~(8) in the Supplementary Information of Ref.~\cite{Bravyi2020}.
The selected positions satisfy $1\leq j<k\leq N$, the active question blocks $\alpha_j$ and $\beta_k$ are independent and uniform on $\{01,10,11\}$, and all other question blocks are $00$.
For the Magic Pentagram problem, the input is uniform on the set $S$ defined in Sec.~4 of Ref.~\cite{Han2022}.
The active question blocks $x_j$ and $y_k$ range independently over $\{000,001,010,011,100\}$, and all other question blocks are $111$.
The notation $j,k$ replaces the positions $k,l$ used in that reference.
In both cases, the selected pair is uniform over all $j<k$.

Let $r=2$ for the Magic Square problem and $r=3$ for the Magic Pentagram problem.
Write $I_j, J_k$ for the sets of bits in the selected question blocks and $O_j,P_k$ for those in the corresponding answer blocks.
For Magic Square, these blocks are $\alpha_j,\beta_k,x_j,y_k$, respectively; for Magic Pentagram, they are $x_j, y_k, z_j, w_k$.
Each block contains $r$ raw circuit bits, and the total number of output bits is $2rN$.
The additional answer components determined by the game parity constraints are not separate circuit outputs.

Following the cited proofs, let $\Gamma(I)$ denote the forward lightcone defined by functional dependence: an output bit belongs to $\Gamma(I)$ if its value can change when a bit in $I$ is flipped with all other input bits fixed.
For either problem, define the event $E$ by
\begin{align} \label{eq:comparison-lightcone-event}
\Gamma(I_j)\cap\Gamma(J_k) & = \emptyset, \\
P_k\cap\Gamma(I_j) & = \emptyset, \\
O_j\cap\Gamma(J_k) & = \emptyset.
\end{align}
This is the event $E_{\mathcal C}$ in Lemma~7 of the Supplementary Information of Ref.~\cite{Bravyi2020}, and the event $E$ preceding Proposition~13 in Appendix~E of Ref.~\cite{Han2022}.
We denote the two events by $E_{\mathrm{MS}}$ and $E_{\mathrm{MP}}$, respectively.
Their definitions depend on the circuit and the selected positions, but not on the active question values.

\subsection{Exceptional probabilities}

The backward lightcone of each output bit has size at most $F^D$.
For a uniform choice of $j<k$, any prescribed pair of input bits belongs to $I_j$ and $J_k$ with probability at most $2/[N(N-1)]$.
Lemma~6(ii) in the Supplementary Information of Ref.~\cite{Bravyi2020}, with $2rN$ output bits, therefore gives
\begin{equation}
\Pr[\Gamma(I_j)\cap\Gamma(J_k)\ne\emptyset] \leq \frac{4rF^{2D}}{N-1}.
\end{equation}

For fixed $j$ and an independent uniform choice of $k$ from $\{1,\ldots,N\}$, Lemma~6(i) of the same reference gives
\begin{equation}
\Pr[O_j\cap\Gamma(J_k)\ne\emptyset] \leq \frac{r2^rF^D}{N}.
\end{equation}
The same estimate holds after averaging over an independent uniform $j$.
The restriction $j<k$ increases this upper bound by at most $N^2/\binom{N}{2}\leq4$.
The same reasoning applies with the two parties interchanged.
A union bound consequently gives
\begin{align}
\Pr[E^c] & \leq \frac{4rF^{2D}}{N-1}+\frac{8r2^rF^D}{N} \\
& \leq \frac{8r(1+2^r)F^{2D}}{N},
\end{align}
where the last inequality uses $N\geq2$ and $F^D\leq F^{2D}$.
For $r=2$, the coefficient is $80$, as in Lemma~7 of the Supplementary Information of Ref.~\cite{Bravyi2020}.
For $r=3$, the same counting argument gives the coefficient $216$ stated in Proposition~13 of Ref.~\cite{Han2022}.
Thus, the exceptional probabilities satisfy
\begin{align} 
q_{\mathrm{MS}} := \Pr[E_{\mathrm{MS}}^c] & \leq \frac{80F^{2D}}{N}, \label{eq:comparison-exceptional-probabilities1}\\
q_{\mathrm{MP}} := \Pr[E_{\mathrm{MP}}^c] & \leq \frac{216F^{2D}}{N}.\label{eq:comparison-exceptional-probabilities2}
\end{align}

\subsection{Conditional success and averaging}

For Magic Square, Lemma~3 in the Supplementary Information of Ref.~\cite{Bravyi2020} gives a necessary condition for a valid circuit output, and Lemma~8 bounds the probability of this condition by $8/9$ for every selected pair in $E_{\mathrm{MS}}$.
The frame parameters need not be independent of the questions.
Instead, the proof of Lemma~8 uses the disjoint lightcones to factor each parameter into a function of Alice's question and a function of Bob's question, and then applies local output relabeling.
For Magic Pentagram, Lemma~14 in Appendix~E of Ref.~\cite{Han2022} states the conditional success bound $19/20$ for $E_{\mathrm{MP}}$ under the input distribution on $S$ specified above.

Let $G\in\{\mathrm{MS},\mathrm{MP}\}$, and let $W_G$ be the event that the classical circuit produces a valid output for problem $G$.
Set $c_{\mathrm{MS}}=8/9$ and $c_{\mathrm{MP}}=19/20$.
The cited conditional bounds imply
\begin{equation}
\Pr[W_G\cap E_G] \leq c_G\Pr[E_G].
\end{equation}
This inequality also holds when $\Pr[E_G]=0$.
Since the success probability on $E_G^c$ is at most one,
\begin{align}
p_G & = \Pr[W_G\cap E_G]+\Pr[W_G\cap E_G^c] \\
& \leq c_G(1-q_G)+q_G \\
& = c_G+(1-c_G)q_G.
\end{align}
Equations~(\ref{eq:comparison-exceptional-probabilities1}) and (\ref{eq:comparison-exceptional-probabilities2}) now yield Eqs.~(\ref{eq:magic-square-refined-classical-bound}) and~(\ref{eq:magic-pentagram-classical-bound}).

The final proof of Theorem~4 in the Supplementary Information of Ref.~\cite{Bravyi2020} and the proof of Theorem~10 in Appendix~F of Ref.~\cite{Han2022} use the weaker estimate $p_G\leq c_G+q_G$.
The bounds used in Sec.~\ref{sec:Qubit} retain the coefficient $1-c_G$ without changing the exceptional events, their probability estimates, or the cited conditional success bounds.

\bibliography{CHSH}

\end{document}